\documentclass{llncs}

\usepackage{amsmath}
\usepackage{amssymb}
\usepackage{stmaryrd}
\usepackage{xstring}

\usepackage{url}

\usepackage{tikz}
\usepackage{color}

\usepackage{enumerate}
\usepackage{proof}

\makeglossary
\newcommand{\IF}{\mbox{{\bf if}\ }}
\newcommand{\FI}{\mbox{{\bf fi}}}
\newcommand{\DO}{\mbox{{\bf do}\ }}
\newcommand{\OD}{\mbox{{\bf od}}}
\newcommand{\WHILE}{\mbox{{\bf while}\ }}
\newcommand{\END}{\mbox{{\bf end}}}
\newcommand{\THEN}{\mbox{\ {\bf then}\ }}
\newcommand{\ELSE}{\mbox{\ {\bf else}\ }}

\newcommand{\SP}{\mbox{\it SP}}

\newcommand{\T}{\mbox{{\bf true}}}

\newcommand{\ES}{\mbox{$\emptyset$}}

\newcommand{\ra}{\mbox{$\:\rightarrow\:$}}

\newcommand{\A}{\mbox{$\ \wedge\ $}}

\newcommand{\sse}{\mbox{$\:\subseteq\:$}}

\newcommand{\te}{\mbox{$\exists$}}

\newcommand{\LL}{\mbox{$\ldots$}}

\newcommand{\ITE}[3]{\mbox{$\IF {#1} \THEN {#2} \ELSE {#3}\ \FI$}}

\newcommand{\WDD}[2]{\mbox{$\WHILE {#1}\ \DO {#2}\ \OD$}}

\newcommand{\HT}[3]{\mbox{$\{{#1}\}\ {#2}\ \{{#3}\}$}}

\newcommand{\sg}{{\mbox{$\sigma$}}}

\newcommand{\BI}[1]{\mbox{$[\![{#1}]\!]_I$}}       
\newcommand{\C}[1]{\mbox{$\{{#1}\}$}}         

\newcommand{\NI}{\noindent} 

\newcommand{\HB}{\qed}

\newcommand{\III}{\vspace{3 mm}}

\def\nlni{\par\ifvmode\removelastskip\fi\vskip\baselineskip\noindent}

\newcommand{\BEGIN}{\mbox{{\bf begin}}}
\newcommand{\block}[1]{\mbox{$\BEGIN \ {#1}\ \END$}}
\newcommand{\local}{\mbox{{\bf local}\ }}

\newenvironment{Example}{\begin{example}}{\end{example}}
\newenvironment{Lemma}{\begin{lemma}}{\end{lemma}}
\newenvironment{Def}{\begin{definition}}{\end{definition}}
\newenvironment{Note}{\begin{note}}{\end{note}}
\newenvironment{Cor}{\begin{corollary}}{\end{corollary}}

\newenvironment{Theorem-HB}{\begin{theorem}}{\HB\end{theorem}}
\newenvironment{Example-HB}{\begin{Example}}{\HB\end{Example}}
\newenvironment{Lemma-HB}{\begin{Lemma}}{\HB\end{Lemma}}
\newenvironment{Def-HB}{\begin{Def}}{\HB\end{Def}}
\newenvironment{Note-HB}{\begin{Note}}{\HB\end{Note}}
\newenvironment{Cor-HB}{\begin{Cor}}{\HB\end{Cor}}
\newenvironment{Warn-HB}{\begin{Warn}}{\HB\end{Warn}}

 \newcommand{\MSI}[1]{\mbox{${\cal M}_I[\![{#1}]\!]$}}

\usepackage{tikz}
\usetikzlibrary{positioning,shapes.geometric}

\title{Reasoning about call-by-value: \\ a  missing result in the history of Hoare's logic}

\author{
Krzysztof R. Apt \\
{CWI, Amsterdam, The Netherlands} \\
{MIMUW, University of Warsaw, Warsaw, Poland} \\[2mm]
Frank S. de Boer \\
{CWI, Amsterdam, The Netherlands} \\
{University of Leiden, The Netherlands}
}
\institute{}%
\begin{document}


\maketitle

\begin{abstract} 
  We provide a sound and relatively complete Hoare-like proof system
  for reasoning about partial correctness of recursive procedures in
  presence of local variables and the call-by-value parameter
  mechanism, and in which the correctness proofs are linear in the length
  of the program. We argue that in spite of the fact that Hoare-like
  proof systems for recursive procedures were intensively studied, no
  such proof system has been proposed in the literature.
\end{abstract}

\section{Introduction}

\subsection{Background and motivation}

Hoare's logic was originally introduced by Tony Hoare in
\cite{Hoa69}. It is the most widely used approach to program
verification.  Over the past fifty years it was successfully extended
to several classes of programs, including parallel and object-oriented
ones, see, e.g. our textbook \cite{ABO09}.  These historical
developments are traced in the recent survey \cite{AO19}.  Also,
formalization of Hoare's logic in various interactive theorem provers,
e.g., \texttt{Coq} (see, e.g., \cite{Chl13}), led to a computer aided
verification of several programs.

One of the crucial features of Hoare's logic is its syntax-oriented
style. It allows one to annotate programs at relevant places (for
example at the entrance of each loop) with \emph{invariants} that
increase programmer's confidence in the correctness of the
program. Also, intended behaviour of procedures can be described by
means of pre- and postconditions, which simplifies program
development.  So it is highly desirable that natural and widely used
programming features can be described in Hoare's logic at an
abstraction level that coincides with the programming language.

The developments that motivate the subject of the present paper passed
through a number of crucial stages.  Already \cite{Hoa71} a proof
system for recursive procedures with parameters was proposed that was
subsequently used in \cite{FH71} to establish correctness of the
\texttt{Quicksort} program. This research got a further impulse thanks
to Stephen A.~Cook who proposed a by now standard notion of
\emph{relative completeness}, since then called \emph{completeness in
  the sense of Cook}, and established in \cite{Coo78} relative
completeness of a proof system for non-recursive procedures.  Cook's
result was extended by Gerald A.~Gorelick in \cite{Gor75}, where a
proof system for recursive procedures was introduced and proved to be
sound and relatively complete.  This line of research led to the
seminal paper \cite{Cla79} of Edmund M.~Clarke who exhibited a
combination of five programming features the presence of which makes
it impossible to obtain a Hoare-like proof system that is sound and
relatively complete.

However, what is often overlooked, all these papers assumed the by now
obsolete call-by-name parameter mechanism.  Our claim is that no paper
so far provided a sound and relatively complete Hoare-like proof
system for a programming language with the following programming
features:

\begin{itemize}
\item a system of mutually recursive procedures,

\item local variables,
  
\item call-by-value parameter mechanism,

\item both dynamic and static scope (notions explained in Section \ref{sec:syntax}).
\end{itemize}
and in which
\begin{itemize}
\item the abstraction level coincides with the programming language,

\item correctness proofs are linear in the length of the programs.
\end{itemize}

Given the above research and the fact that call-by-value is in many
programming languages the main parameter mechanism, this claim may
sound surprising.  Of course, there were several contributions to
Hoare's logic concerned with recursive procedures, but ---as we
explain in Section \ref{sec:final}--- none of them provided a
proof system that met the above criteria.  The aim of this paper is to
provide such a Hoare-like proof system.

One of the notorious problems when dealing with the above programming
features is that variables can occur in a program both as local and
global.  The way this double use of variables is dealt with has direct
consequences on what is being formalized. In particular, the proof
systems studied in \cite{Coo78} and \cite{Gor75} dealt with dynamic
scope and not static scope.  Additionally, variables can be used as
formal parameters and can occur in actual parameters.  This multiple
use of variables can lead to various subtle errors and was usually
dealt with by imposing some restrictions on the actual parameters, notably that the variables of the actual
parameters cannot be changed by the call. In
our approach no such restrictions are present but these complications
explain the seemingly excessive care we exercise when dealing with
this matter.

\subsection{Plan of the paper}

In the next section we introduce a small programming language we are
dealing with in the paper and identify a natural subset of clash-free
programs for which dynamic and static scope coincide.
Next, in Section \ref{sec:semantics}, we recall various
aspects of semantics introduced in \cite{ABO09} and establish some
properties that are used in the sequel. Then, in Section
\ref{sec:towards}, we provide some insights that motivate the crucial
proof rules of the proof system that is introduced in Section
\ref{sec:proof}.  These insights suggest that the adopted recursion
rule is derivable in a (sound) proof system of \cite{ABO09}. We prove
this fact in Section \ref{sec:pure}, which allows us to conclude that the
proof system given in Section \ref{sec:proof} is sound.

Next, in Section \ref{sec:sandc}, we establish relative completeness
of the proof system and in Section \ref{sec:length} analyze the length
of proofs in this proof system. In the final section we discuss
related work and summarize our approach.

\section{Syntax}
\label{sec:syntax}

Throughout the paper we assume a fixed first-order language $\cal L$.
\emph{Expressions} are terms in the language $\cal L$, \emph{Boolean
  expressions} are quantifier-free formulas of $\cal L$, while
\emph{assertions} are formulas of $\cal L$.  

We denote the set of all variables of $\cal L$ by \emph{Var} and for a
sequence $\bar{t}$ of expressions denote the set of all variables
occurring in $\bar{t}$ by $var(\bar{t})$.  The set of variables that
occur free in an assertion $p$ is defined in a standard way and
denoted by $free(p)$.

A (simultaneous) \emph{substitution} of a list of expressions $\bar{t}$
for a list of distinct variables $\bar{x}$ of the same length is
written as $[\bar{x}:=\bar{t}]$ and the result of applying it to an
expression or an assertion $s$ as $s[\bar{x}:=\bar{t}]$.  To ensure
a uniform presentation we allow the empty substitution in which
the list of variables $\bar{x}$ is empty.

We now move on and introduce the syntax of the programs.  For
simplicity in the considered toy programming language we admit only
simple variables (so no array or subscripted variables), all of the
same type.  \emph{Statements} are defined by the following grammar:
\[
\begin{array}{ll}
  S::= & skip \mid
         \bar{x}:=\bar{t} \mid
        P(\bar{t}) \mid
         S; \ S \mid 
         \ITE{B}{S}{S} \mid 
  \\[1mm]
      &   \WDD{B}{S}
          \mid
        \block{\local \bar{x}:=\bar{t}; \ S},
\end{array}
\]
where
\begin{itemize}
\item $\bar{x}:=\bar{t}$ is a \emph{parallel assignment}, with
  $\bar{x}$ a (possibly empty) list of distinct variables and
  $\bar{t}$ a list of expressions of the same length as $\bar{x}$;
  when $\bar{x}$ is empty we identify $\bar{x}:=\bar{t}$ with the $skip$ statement,
  
\item $P$ is a procedure name; each procedure $P$ is defined by a declaration of the form
\[
  P(\bar{u})::S,
\]
where $\bar{u}$ is a (possibly empty) list of distinct variables,
called \emph{formal parameters} of the procedure $P$, and $S$ is a
statement, called the \emph{body} of the procedure $P$,

\item $P(\bar{t})$ is a procedure call, with the \emph{actual
    parameters} $\bar{t}$, which is a (possibly empty) list of
  expressions of the same length as the corresponding list of formal
  parameters,
  
\item $B$ is a Boolean expression,

\item $\block{\local \bar{x}:=\bar{t}; \ S}$ is a \emph{block statement}
where $\bar{x}$ is a (possibly empty) list of distinct 
variables, all of which are explicitly initialized by means of the
parallel assignment $\bar{x}:=\bar{t}$.

\end{itemize}

Of special interest in our approach will be \emph{generic calls} which
are calls in which the actual parameters coincide with the formal
ones.

By a \emph{program} we mean a pair $(D \mid S)$, where $S$ is a
statement, called the \emph{main statement} and $D$ is a set of
procedure declarations such that each procedure (name) that appears in
$S$ or $D$ has exactly one procedure declaration in $D$.  So we allow
mutually recursive procedures but not nested procedures.  In what
follows we assume that the set of procedure declarations $D$ is given
and fixed and whenever no confusion arises we omit the references to it
when discussing specific (main) statements.
We denote by $var(D \mid S)$ the set of variables that occur in 
$(D \mid S)$.

The parallel assignment plays a crucial role in the way procedure
calls are dealt with: the procedure call $P(\bar{t})$, where $P$ is
declared by $P(\bar{u})::S$, is interpreted as the block statement
$\block{\local \bar{u}:=\bar{t}; \ S}$, where $\bar{u}:=\bar{t}$
models the parameter passing by value and the block statement ensures
that the changes to the formal parameters $\bar{u}$ are local. Such a
replacement of a procedure call by an appropriately modified procedure
body is called \emph{inlining} or a \emph{copy rule}.

In our setup so interpreted inlining results in \emph{dynamic scope},
which means that each procedure call is evaluated in the environment
in which it is called.  The simplest example is the parameterless
procedure $P$ declared by $P() :: y:=x$ and the main statement
$x:=0;\ \block{\local x:=1; \ P()}$. Here the inlining results in the
program
\[
  x:=0;\ \block{\local x:=1; \ \block{\local skip; \ y:=x}}
\]
that yields $y=1$ upon termination. However, if we renamed the
occurrence of $x$ in the block statement to a fresh variable, say
$x'$, and thus used the statement
$x:=0;\ \block{\local x':=1; \ P()}$, then inlining would result in
the program
\[
  x:=0;\ \block{\local x':=1; \ \block{\local skip; \ y:=x}}
\]
that yields $y=0$ upon termination. In the latter situation dynamic
scope coincides with the \emph{static scope}, which means that each
procedure call is evaluated in the environment in which it is
declared.

Above example shows that static scope can be ensured
when certain variable name clashes are avoided.
This can be made precise as follows.

\begin{definition}
\mbox{} \\[-3mm]
  \begin{itemize}
  \item An occurrence of a variable in a statement $S$ is called
  \emph{local} if it occurs inside a block statement
  $\block{\local \bar{x}:=\bar{t}; \ S_1}$ within $S$, either in the
  list $\bar{x}$ or in $S_1$.

\item An occurrence of a variable in a procedure declaration
  $ P(\bar{u})::S$ is called \emph{local} if it is in the
  list $\bar{u}$ or its occurrence in $S$ is local.

\item An occurrence of a variable in a program $(D \mid S)$
  is called \emph{local} if it is a local occurrence in $S$ or
  in one of the procedure declarations in $D$.

\item An occurrence of a variable in a program is called \emph{global} if it is not
  local.
\HB
\end{itemize}

\end{definition}

The following
example should clarify this definition.

\begin{example} \label{exa:nested}
  Consider the nested block statement
  \[
    S \equiv \block{\local u,y:=0,u+1; \ \block{\local u:=u+2; \ P(u+y)}}.
  \]
Out of 5 occurrences of $u$ in it only the one in $u+1$ is global.
In turn, both occurrences of $y$ are local.

Further, in the procedure declaration
\[
  P(u) :: \block{\local u:=x; \ y:=x}; \ u:=1
\]
all occurrences of $u$ are local, while all occurrences of $x$ and
$y$ are global.
\HB
\end{example}

Note that in the statement $\block{\local \bar{x}:=\bar{t}; \ S}$ the
occurrences of the variables from $\bar{x}$ in $\bar{t}$ are not
considered as local. (They can be local if this block statement is enclosed
by another block statement in which the same variable is declared.)
Thanks to this a global  variable occurring in 
a procedure call $P(\bar{t})$ remains global after inlining takes place.
Informally, such global occurrences of variables are not 'captured' by the
block statement used to model inlining, and thus no renaming of variables is needed.

We now introduce an important class of programs.

\begin{definition} \label{def:clash-free}
  A program is called \emph{clash-free} if no variable has both a
  local occurrence in it and a global occurrence in a procedure body.
\end{definition}

The following observation shows that clash-freeness is preserved by
inlining.

\begin{note}
  Suppose that $(D \mid S_1)$ is a clash-free program and that $S_2$
  results from $S_1$ by replacing a procedure call $P(\bar{t})$ by
  $\block{\local \bar{u}:=\bar{t}; \ S}$, where
  $P$ is declared by $P(\bar{u})::S$.  Then $(D \mid S_2)$ is
  clash-free.
\end{note}

\begin{proof}
  The programs $(D \mid S_2)$ and $(D \mid S_1)$ use the same
  variables.  Each variable $x$ that has a local occurrence in $S_2$
  has either a local occurrence in $S_1$ or a local occurrence in $S$
  or is an element of the list $\bar{u}$.  In all three cases these
  are local occurrences in $(D \mid S_1)$, so $x$ does not have a
  global occurrence in a procedure body.
  \HB
\end{proof}

Informally, this observation states that inlining preserves the
separation between the local variables of the program and the global
variables of the procedures.  This means that the restriction to
clash-free programs guarantees that their dynamic and static scope interpretations coincide.
So a programmer in the considered programming language can ensure
static scope by adhering to a simple syntactic convention.


In our considerations it will be important to refer to the set of
variables that a given program can change. Of course, this set depends
on the initial state in which the program is executed but it suffices to use
its approximation from above. This can be done in a number of ways that lead
to possibly different outcomes.
For procedure calls, for example,  it would be natural to put
\[
change(D\mid P(\bar{t})) := change(D\mid S) \setminus \{\bar{u}\}, 
\]
where $P$ is declared by $P(\bar{u})::S$, but such a definition would
be circular in the presence of recursive calls. 
This could be rectified by an appropriate fixpoint definition.
In our setting the following more direct
definition will do.
Given a program $(D \mid S)$ we define inductively the set of variables
that can be changed by it as follows:

\begin{itemize}

\item $change(D \mid skip) := \ES$,
  
\item $change(D \mid \bar{x}:=\bar{t}) := \{\bar{x}\}$,

\item
  \[change(D \mid P(\bar{t})) :=
    \begin{cases}
    change(D) \setminus \{\bar{u}\} & \textrm{ if }  P(\bar{u})::S \in D \\
    \ES & \textrm{ otherwise}
  \end{cases}
\]


\item $change(D \mid S_1 ; \ S_2) := change(D \mid S_1) \cup change(D \mid S_2)$,

\item $change(D \mid \ITE{B}{S_1}{S_2}) := change(D \mid S_1) \cup change(D \mid S_2)$,  

\item $change(D \mid \WDD{B}{S}) := change(D \mid S)$,

\item $change(D \mid \block{\local \bar{x}:=\bar{t}; \ S}) := change(D \mid S) \setminus \{\bar{x}\}$,

\item $change(P(\bar{u})::S) := change(\ES \mid S)$,

\item $change(\{P(\bar{u})::S\} \cup D) := change(P(\bar{u})::S) \cup change(D)$.



\end{itemize}

In the above definition the crucial clause is the one concerning the procedure calls. The first
alternative 
formalizes the idea that the formal parameters are local
w.r.t.~procedure body and as a result they cannot change during the
execution of the call. The second alternative is needed to deal with
the procedure calls in the definition of $change(P(\bar{u})::S)$.

\section{Semantics}
\label{sec:semantics}

In this section we gather various facts concerning semantics. We begin
by a slightly adjusted presentation extracted from \cite{ABO09},
followed by a collection of various properties that will be needed
later.

As already mentioned, we assume for simplicity that all variables are
of the same type, say $T$. Each realization of this type (called a
\emph{domain)} yields an \emph{interpretation} $I$ that assigns a
meaning to the function symbols and relation symbols of the language
$\cal L$.

Fix an intepretation $I$. By a \emph{state} we mean a function that
maps each variable to an element of the domain of $I$.  We denote the
set of states by $\Sigma_I$.

An \emph{update} $\sigma[x:=d]$ of a state $\sigma$, where $x$ is a
variable and $d$ an element of the domain of $I$, is a state that
coincides with $\sigma$ on all variables except $x$ to which it
assigns $d$.  An \emph{update} $\sigma[\bar{x}:=\bar{d}]$ of a state
$\sigma$, where $\bar{x}$ is a list of variables and $\bar{d}$ is a list
of domain elements of the same length, is defined analogously.

Given a state $\sigma$ and an expression $t$ we define the element $\sigma(t)$ of
the domain of $I$ inductively in the standard way and for $\bar{t} := t_1, \LL, t_k$ we put
$\sigma(\bar{t}) := (\sigma(t_1), \LL, \sigma(t_k))$. 

For a set $Z$ of variables we denote by $\sg[Z]$ the \emph{restriction} of the
state $\sg$ to the variables occurring in $Z$ and write for two states $\sg$ and $\tau$
\[
\sg = \tau \ {\bf mod} \ Z  
\]
if $\sg[\mathit{Var}\setminus Z] = \tau[\mathit{Var}\setminus Z]$.
We extend the definitions of an update and equality ${\bf mod} \ Z$
 to, respectively,
a set of states and sets of states in the expected way.

The relation `assertion $p$ is true in the state $\sigma \in \Sigma_I$', denoted by
$\sigma \models_I p$, is defined inductively in the standard way.  
The \emph{meaning} of an assertion w.r.t.~the interpretation $I$,
written as $\BI{p}$, is defined by
\[ 
 \BI{p}=\C{\sg \in \Sigma_I \mid \sg \models_I p}. 
\]
We say that $p$ is \emph{true in $I$}, and write $\models_I p$, if
$\BI{p}=\Sigma_I$.

The \emph{meaning} of a program $(D \mid S)$ is a function
\[
  \MSI{S}: \Sigma_I \to {\cal P}(\Sigma_I)
\]
the definition of which, given in \cite{ABO09}, is omitted.

For a given state $\sg$, $\MSI{S}(\sigma) = \{\tau\}$ states the fact
that the program $(D \mid S)$ terminates when started in the initial
state $\sigma$, yielding the final state $\tau$.  If $(D \mid S)$ does
not terminate when started in $\sigma$, then $\MSI{S}(\sg)$ is the
empty set.

We extend the function $\MSI{S}$ to deal with sets of states $X \sse \Sigma_I$ by
\[ 
\MSI{S}(X)= \bigcup_{\sg \in X} \MSI{S}(\sg).
\]

The following lemma collects various consequences of the semantics
given in \cite{ABO09} that will be used below.

\begin{lemma} \label{lem:sem}
\mbox{} \\[-3mm]
\begin{description}
  
\item[Skip] For all states $\sigma$
  \[
    \MSI{skip}(\sigma) = \{\sigma\}.
  \]

\item[Assignment] For all states $\sigma$
  \[
    \MSI{\bar{x}:=\bar{t}}(\sigma) = \{\sigma[\bar{x}:=\sigma(\bar{t})]\}.
  \]

\item[Composition] For all states $\sigma$
  \[
    \MSI{S_1; \ S_2}(\sigma) = \MSI{S_2}(\MSI{S_1}(\sigma)).
  \]
  
  \item[Block] For all states $\sigma$
    \[
      \MSI{\block{\local \bar{x}:=\bar{t}; \ S}}(\sigma) = \MSI{\bar{x}:=\bar{t}; \ S}(\sigma)[\bar{x}:=\sigma(\bar{x})].
    \]
    
      \item[Inlining]
For a procedure $P$ declared by $P(\bar{u})::S$
\[
  \MSI{P(\bar{t})} = \MSI{\block{\local \bar{u}:=\bar{t}; \ S}}.
\]

\item[Access and Change] For all states $\sigma$ and $\tau$
  if $\MSI{S}(\sg) = \{\tau\}$ then
\[
\tau[\mathit{Var}\setminus change(D \mid S )] = \sigma[\mathit{Var}\setminus change(D \mid S )].
\]

%
%
\end{description}
\end{lemma}

The \textbf{Access and Change} item 
formalizes the intuition that when executing a program
$(D \mid S)$ only variables in $change(D \mid S )$ can be modified.
It can be equivalently stated as
that $\MSI{S}(\sg) \neq \ES$ implies
    \[
      \MSI{S}(\sg) = \{\sg\}\ {\bf mod} \ change(D \mid S ).
    \]


\begin{corollary} \label{cor:conc}
Suppose that the procedure $P$ is declared by $P(\bar{u})::S$.
\begin{enumerate}[(i)]

\item 
  For all states $\sigma$ 
\[
  \MSI{P(\bar{t})}(\sg) =
  \MSI{S}(\sigma[\bar{u}:=\sigma(\bar{t})])[\bar{u}:=\sigma(\bar{u})].
\]

\item $\MSI{P(\bar{t})} = \MSI{\block{\local \bar{u}:=\bar{t}; \ P(\bar{u})}}$.

\end{enumerate}
\end{corollary}

\begin{proof}
\mbox{} 

\NI
$(i)$ By the appropriate items of Lemma \ref{lem:sem} we successively have
\[
\begin{array}{l}
  \MSI{P(\bar{t})}(\sigma) = \\
  \MSI{\block{\local \bar{u}:=\bar{t}; \ S}}(\sigma) = \\
  \MSI{\bar{u}:=\bar{t}; \ S}(\sigma)[\bar{u}:=\sigma(\bar{u})] = \\
  \MSI{S}(\sigma[\bar{u}:=\sigma(\bar{t}])[\bar{u}:=\sigma(\bar{u})].
\end{array}
\]
  
\NI
$(ii)$  Take a state $\sigma$. Denote $\sigma[\bar{u}:=\sigma(\bar{t})]$ by
$\sigma'$. By the appropriate items of Lemma \ref{lem:sem} and $(i)$ we
successively have
\[
\begin{array}{l}
  \MSI{\block{\local \bar{u}:=\bar{t}; \ P(\bar{u})}}(\sigma) = \\
  \MSI{\bar{u}:=\bar{t}; \ P(\bar{u})}(\sigma)[\bar{u}:=\sigma(\bar{u})] = \\
  \MSI{P(\bar{u})}(\sigma')[\bar{u}:=\sigma(\bar{u})] = \\
  \MSI{S}(\sigma'[\bar{u}:=\sigma'(\bar{u})])[\bar{u}:=\sigma(\bar{u})] = \\
  \MSI{S}(\sigma')[\bar{u}:=\sigma(\bar{u})] = \\
  \MSI{P(\bar{t})}(\sigma).
\end{array}
\]
\HB
\end{proof}

  


By a \emph{correctness formula} we mean a triple $\HT{p}{S}{q}$, where $p, q$
are assertions and $S$ is a statement. Given a program
$(D \mid S)$, a correctness formula $\HT{p}{S}{q}$, and an
interpretation $I$, we write
\[
  I \models \HT{p}{S}{q}
\]
if
\[
  \MSI{S}(\BI{p}) \sse \BI{q}.
\]
We say then that $\HT{p}{S}{q}$ is true in $I$ \emph{in the sense of partial correctness}.


\begin{lemma} \label{lem:q}
  Suppose that for all states $\sg$
\[
  \MSI{S}(\sg) = \MSI{T}(\sigma) \ {\bf mod} \ \{\bar{u}\}.
\]
Then for all assertions $p$ and $q$ such that $\C{\bar{u}} \cap free(q) = \ES$
\[
\mbox{$I \models \HT{p}{S}{q}$ iff $I \models \HT{p}{T}{q}$.}
\]
\end{lemma}
\begin{proof}
By the first assumption
  \begin{equation}
    \label{equ:BI}
 \MSI{S}(\BI{p}) =  \MSI{T}(\BI{p}) \ {\bf mod} \ \{\bar{u}\}.    
  \end{equation}

  By the second assumption
  $free(q) \sse \textit{Var} \setminus \C{\bar{u}}$, so for arbitrary
  states $\sigma$ and $\tau$ such that
  $\sigma = \tau \ {\bf mod} \ \C{\bar{u}}$ we have
  $\sigma \models_I q$ iff $\tau \models_I q$.  (This is actually the
  Coincidence Lemma in \cite[p.~47]{ABO09}.)  Hence for two sets of
  states $X$ and $Y$ such that $X=Y\ {\bf mod} \ \C{\bar{u}}$ we have
\[
  X \sse \BI{q} \mathrm{\ iff \ }   Y \sse \BI{q}.
\]
So the desired equivalence follows by (\ref{equ:BI})
and the definition of $I \models \HT{p}{S}{q}$.
\HB
\end{proof}

\begin{corollary} \label{cor:ht}
For all assertions $p$ and $q$ such that $\C{\bar{u}} \cap free(q) = \ES$

  \begin{enumerate}[(i)]
  \item 
$I \models \HT{p}{\block{\local \bar{u}:=\bar{t}; \ S}}{q}$ iff $I \models \HT{p}{\bar{u}:=\bar{t}; \ S}{q}$.
  
\item
  $I \models \HT{p}{P(\bar{t})}{q}$ iff $I \models \HT{p}{\bar{u} := \bar{t}; \ S}{q}$,
 
where the procedure $P$ is declared by $P(\bar{u})::S$.
  \end{enumerate}

\end{corollary}

\begin{proof}
By the \textbf{Block} and \textbf{Inlining} items of Lemma \ref{lem:sem}
\[
    \MSI{\block{\local \bar{u}:=\bar{t}; \ S}}(\sg) =  \MSI{\bar{u}:=\bar{t}; \ S}(\sg) \ {\bf mod} \ \{\bar{u}\}
\]
and

\[
  \MSI{P(\bar{t})}(\sg) = \MSI{\bar{u}:=\bar{t}; \ S}(\sigma) \ {\bf mod} \ \{\bar{u}\},
\]
so the claim follows by Lemma \ref{lem:q}.
  \HB
\end{proof}

Following \cite{Coo78} we now introduce the following notions. Denote
by $sp_I(p,S)$ the \emph{strongest postcondition} of a program
$(D \mid S)$ w.r.t.~an assertion $p$, defined by
\[
sp_I(p,S) = \MSI{S}(\BI{p}).
\]
So $sp_I(p,S)$ is the set of states that can be reached by executing
$(D \mid S)$ starting in a state satisfying $p$.

We say that a set of states $\Sigma$ is \emph{definable} in an
interpretation $I$ iff for some formula $\phi$ of $\cal L$ we have
$\Sigma = \BI{\phi}$. We say then that $\phi$ \emph{defines} $\Sigma$.
Further, we say that the language $\cal L$ is \emph{expressive}
relative to an interpretation $I$ if for every assertion $p$ and
program $(D \mid S)$ the set of states $sp_I(p,S)$ is definable in
$I$. In that case we denote by $\SP_I(p,S)$ a formula that defines the
set $sp_I(p,S)$.

Consider a proof system $PS$ for proving correctness formulas. 
Given a (possibly empty) set of correctness formulas $\Phi$ and an interpretation $I$ we write
\[
\Phi \vdash_{{PS}, I} \HT{p}{S}{q}
\]
to denote the fact that $\HT{p}{S}{q}$ can be proved in $PS$ from
$\Phi$ and the set of all assertions true in $I$ (that can be used as
premises in the CONSEQUENCE rule introduced in Section \ref{sec:proof}), and
omit $\Phi$ if it is empty.  We say that

\begin{itemize}
\item $PS$ is \emph{sound} if for every interpretation $I$ and
  correctness formula $\HT{p}{S}{q}$,
  \[
    \mbox{$\vdash_{{PS}, I} \HT{p}{S}{q}$ implies $I \models \HT{p}{S}{q}$,}
  \]

\item $PS$ is \emph{relatively complete}, or \emph{complete in the
    sense of Cook}, if for every interpretation $I$ such that $\cal L$
  is expressive relative to $I$ and correctness formula
  $\HT{p}{S}{q}$,
\[
  \mbox{$I \models \HT{p}{S}{q}$ implies $\vdash_{{PS}, I} \HT{p}{S}{q}$.}
\]
\end{itemize}

\section{Towards proofs of linear length}
\label{sec:towards}

To reason about correctness of programs defined in Section \ref{sec:syntax}
we used in \cite{ABO09} the following recursion rule:
\III

\NI
\[
\begin{array}{l}
\HT{p_1}{P_1(\bar{t}_1)}{q_1},\ldots,\HT{p_k}{P_k(\bar{t}_k)}{q_k} \vdash \HT{p}{S}{q},                    \\
\HT{p_1}{P_1(\bar{t}_1)}{q_1},\ldots,\HT{p_k}{P_k(\bar{t}_k)}{q_k} \vdash \\
\qquad \HT{p_i}{\block{\local \bar{u}_i:=\bar{t}_i;\  S_i}}{q_i}, \ i \in \{1, \LL, k\} \\
[-\medskipamount]
\hrulefill                                                      \\
\HT{p}{S}{q} 
\end{array}
\]
where $P_i(\bar{u}_1) ::S_i\in D$ for $i \in \{1, \LL, k\}$ and
$P_1(\bar{t}_1), \LL, P_n(\bar{t}_k)$ are all procedure calls that appear in $(D \mid S)$.
The $\vdash$ sign refers to the provability
using the remaining axioms and proof rules. 
\III

So it is allowed here that $P_i\equiv P_j$ for some $i \neq j$.  In
this rule there are $k+1$ subsidiary proofs in the premises, where $k$
is the total number of procedure calls that appear in $(D \mid S)$.
Note that the statements used on the right-hand sides of the last $k$
provability signs $\vdash$ are the corresponding effects of inlining
applied to the procedure calls on the left-hand side of $\vdash$.  In
this proof rule each procedure calls requires a separate subsidiary
correctness proof.  This results in inefficient correctness proofs.

More precisely, assuming a program $(D \mid S)$ with $k$ procedure
calls, each of the $k+1$ subsidiary proofs in the premises of the
above recursion rule can be established in the number of steps linear
in the length of $(D \mid S)$.  But $k$ is linear in the length of
$(D \mid S)$, as well, and as a result the bound on the length of the
whole proof is quadratic in the length of $(D \mid S)$.  This bound
remains quadratic even for programs with a single procedure, since $k$
remains then linear in the length of $(D \mid S)$.

Can we do better? Yes we can, by proceeding through a couple of simple
steps. First, we replace each  procedure call $P(\bar{t})$
such that $\bar{t} \neq \bar{u}$, where $\bar{u}$ are the formal parameters of $P$,
by the block statement $\block{\local \bar{u}:=\bar{t}:P(\bar{u})}$.
This gives rise to so-called \emph{pure
  programs} that only contain generic procedure calls (as introduced in Section \ref{sec:syntax}).  For pure
programs the last $k$ premises of the above recursion rule reduce to
\[
\begin{array}{l}
\HT{p_1}{P_1(\bar{u}_1)}{q_1},\ldots,\HT{p_k}{P_k(\bar{u}_k)}{q_k} \vdash \\
\qquad \HT{p_i}{\block{\local \bar{u}_i:=\bar{u}_i;\  S_i}}{q_i}, \ i \in \{1, \LL, k\}.
\end{array}
\]

The transformation of programs the into pure ones can be avoided as follows.
Incorporating in the proof system the BLOCK rule that will be introduced below,
the above $k$ premises can be reduced to
\[
\begin{array}{l}
\HT{p_1}{P_1(\bar{u}_1)}{q_1},\ldots,\HT{p_k}{P_k(\bar{u}_k)}{q_k} \vdash \\
\qquad \HT{p_i}{S_i}{q_i}, \ i \in \{1, \LL, k\} \\
\end{array}
\]
provided that $\C{\bar{u}_i} \cap free(q_i) = \ES$ for
$i\in\{1,\ldots,k\}$.

Further, the replacement of every non-generic procedure call in the
considered program by its corresponding block statement that uses a
generic call can be captured proof theoretically by means of the proof
rule
\[
\frac{\HT{p}{\block{\local \bar{u}:=\bar{t}:P(\bar{u})}}{q}}
{\HT{p}{P(\bar{t})}{q}}
\]
This rule can be further simplified to a direct instantiation of the
generic calls, using the BLOCK rule. 

Finally, recall that $k$ is the number of different procedure calls
that appear in $(D \mid S)$. But for pure programs $k \leq n$ (or
$k = n$ if each procedure declared is also called).  The resulting
proof system is described in detail in the following section.

\section{Proof system}
\label{sec:proof}

As before the set of declarations $D$ is given and fixed, so we omit the
references to it whenever possible.  We assume a usual proof system
concerned with the correctness formulas about all statements except
the block statements and the procedure calls.

We  use the following recursion rule that refers to generic calls:
\III

\NI
RECURSION 
\[
\begin{array}{l}
\HT{p_1}{P_1(\bar{u}_1)}{q_1},\ldots,\HT{p_n}{P_n(\bar{u}_n)}{q_n} \vdash \HT{p}{S}{q},                    \\
\HT{p_1}{P_1(\bar{u}_1)}{q_1},\ldots,\HT{p_n}{P_n(\bar{u}_n)}{q_n} \vdash \HT{p_i}{S_i}{q_i}, \ i \in \{1, \LL, n\} \\
[-\medskipamount]
\hrulefill                                                      \\
\hspace*{45mm} \HT{p}{S}{q}
\end{array}
\]
where $D = \{P_i(\bar{u}_i) ::S_i \mid i\in\{1,\ldots,n\}\}$
and $\C{\bar{u}_i} \cap free(q_i) = \ES$ for
$i\in\{1,\ldots,n\}$, and the $\vdash$ sign refers to the provability
using the remaining axioms and proof rules. 
\III

In the case of just one procedure this rule simplifies to the following one, in which
one draws a conclusion only about the generic procedure call:
\III

\NI
RECURSION I
\[
  \frac{\HT{p}{P(\bar{u})}{q} \vdash \HT{p}{S}{q}}
  {\HT{p}{P(\bar{u})}{q}}
\]
where $D = \{P(\bar{u}) ::S\}$
and $\C{\bar{u}} \cap free(q) = \ES$.

\III

Generic calls can be  instantiated using the following rule:
\III

\NI
PROCEDURE CALL
\[
\frac{\HT{p}{P(\bar{u})}{q}}
{\HT{p[\bar{u}:=\bar{t}]}{P(\bar{t})}{q}}
\]
where $P(\bar{u}) ::S\in D$ for some $S$ and
$\C{\bar{u}} \cap free(q) = \ES$.
\III

For block statements we use the following rule from \cite[p.~158]{ABO09}:
\III

\NI
BLOCK
\[
\frac{\HT{p}{\bar{x} := \bar{t};\  S}{q}}
{\HT{p}{\block{\local \bar{x} := \bar{t};\  S}}{q}}
\]
where $\C{\bar{x}} \cap free(q) = \ES$.
\III

We shall also need the following substitution rule from
\cite[p.~98]{ABO09} to deal with the occurrences in the assertions of
local variables of block statements and the formal parameters of
the procedure calls:
\III

\NI
SUBSTITUTION 
\[ \frac{ \HT{p}{S}{q} }  
        { \HT{p[\bar{x}:=\bar{y}]}{S}{q[\bar{x}:=\bar{y}]} } 
\]
where $\{\bar{x}\} \cap var(D \mid S )=\ES$ and
$\{\bar{y}\} \cap change(D \mid S )=\ES$.
\III


Additionally, we use the following proof rules also used in
\cite{Gor75} (though the side conditions are slightly different):
\III

\NI
INVARIANCE

\[ \frac{ \HT{r}{S}{q}           }
        { \HT{p \A r}{S}{p \A q} } 
\]
where $free(p) \cap change(D \mid S )=\ES$.
\III

\NI
$\te$-INTRODUCTION
\[ \frac{ \HT{p}{S}{q}          }
        { \HT{\te \bar{x}:p}{S}{q}    }
\]
where $\{\bar{x}\} \cap (var(D \mid S ) \cup free(q)) = \ES$. 
\III

Finally, we shall need the following rule from \cite{Hoa69}:
\III

\NI
CONSEQUENCE
\[
 \frac{ p \ra p_1, \HT{p_1}{S}{q_1}, q_1 \ra q        }
        { \HT{p}{S}{q}  }
\]
\III

We denote the above proof system by \emph{CBV} (for call-by-value) and write from now on
$\Phi \vdash_I \HT{p}{S}{q}$ instead of $\Phi \vdash_{CBV, I} \HT{p}{S}{q}$.

\begin{example} \label{exa:1}
  We provide here two representative proofs in our proof system.

\NI
$(i)$ Consider the block statement
$\block{\local u:=t; \ x:=u}$. Suppose that $x \not \in var(t)$. We prove that then 
\begin{equation}
\HT{\T}{\block{\local u:=t; \ x:=u}}{x=t}.  
\label{equ:block}
\end{equation}

By the standard ASSIGNMENT axiom and the COMPOSITION and CONSEQUENCE
rules of \cite{Hoa69} we get on the account of the assumption about
$x$ and $t$
\[
\HT{\T}{u:=t; \ x:=u}{x=t}.
\]
But we cannot now apply the BLOCK rule because $u$ can occur in $t$. So we 
introduce a fresh variable $u_0$ and proceed as follows.
Let $t' = t[u:=u_0]$. 

By the ASSIGNMENT axiom
\[
    \HT{t = t'}{u:=t}{u = t'}
\]
and
\[
  \HT{u = t'}{x:=u}{x=t'},
\]
so by the COMPOSITION rule
\[
  \HT{t = t'}{u:=t; \ x:=u}{x=t'}.
\]
Now we can apply the BLOCK rule since $u \not \in free(x=t')$ which yields
\[
  \HT{t = t'}{\block{\local u:=t; \ x:=u}}{x=t'}
\]
Using the SUBSTITUTION rule with the substitution $[u_0:=u]$
we then get (\ref{equ:block}) by the CONSEQUENCE rule 
since $(t = t')[u_0:=u]$ and $x=t'[u_0:=u] \to x = t$
are true.


  \III

  \NI
$(ii)$ To illustrate the use of the PROCEDURE
  CALL rule consider the following example.  We shall return to it in
  Section \ref{sec:final}.

  Assume the following procedure declaration:
\[
add(u):: sum:=sum + u
\]
and consider the following correctness formula:
\begin{equation}
  \label{equ:sum}
\HT{sum=z}{add(sum)}{sum=z+z}.  
\end{equation}
So the variable $sum$ is here both a global variable and an actual parameter of the
procedure call.

We prove (\ref{equ:sum}) as follows.
First, we get by the ASSIGNMENT
\[
\HT{sum+u=z+v}{sum:=sum+u}{sum=z+v},
\]
so by the CONSEQUENCE rule
\[
\HT{sum=z \land u=v}{sum:=sum+u}{sum=z+v},
\]
from which
\[
\HT{sum=z \land u=v}{add(u)}{sum=z+v}
\]
follows by the simplified version of the RECURSION I rule.
Applying the PROCEDURE CALL rule we get
\[
\HT{sum=z \land sum=v}{add(sum)}{sum=z+v}.
\]

Next, applying the SUBSTITUTION rule with the substitution $[v:=z]$ we obtain
\[
\HT{sum=z \land sum=z}{add(sum)}{sum=z+z},
\]
from which (\ref{equ:sum}) follows by a trivial application of the CONSEQUENCE rule.
\HB
\end{example}

\section{Soundness}
\label{sec:pure}

Considerations of Section \ref{sec:towards} suggest that our proof
system \emph{CBV} is closely related to the proof system \emph{ABO} of
\cite{ABO09}.  We now exploit this observation to establish 
soundness of \emph{CBV} in a more informative way.
Recall that it was proved, respectively in \cite{ABO09} and \cite{ABOG12}, 
that \emph{ABO} is sound and relatively complete.  (This result
was in fact a stepping stone in a proof of soundness and relative
completeness of a related proof system for a simple class of
object-oriented programs.)

We begin by relating the RECURSION rule to the \emph{ABO} proof
system.  Call the recursion rule used in \cite{ABO09} and discussed in
Section \ref{sec:towards} RECURSION II rule.  Recall that a program is
pure if all procedure calls in it are generic.

\begin{lemma} \label{lem:pure}
  For pure programs the RECURSION rule is a derived rule in the proof system
\textit{ABO}.
\end{lemma}

\begin{proof}
  Fix an interpretation $I$. Suppose that we established in the proof
  system \emph{ABO} the premises of the RECURSION rule, i.e.,
\begin{equation}
  \label{equ:ass1}
\HT{p_1}{P_1(\bar{u}_1)}{q_1},\ldots,\HT{p_n}{P_n(\bar{u}_n)}{q_n} \vdash_{{ABO}, I} \HT{p}{S}{q}  
\end{equation}
and
\begin{equation}
  \label{equ:ass2}
\HT{p_1}{P_1(\bar{u}_1)}{q_1},\ldots,\HT{p_n}{P_n(\bar{u}_n)}{q_n} \vdash_{{ABO}, I} \HT{p_i}{S_i}{q_i}  
\end{equation}
for $i \in \{1, \LL, n\}$, where
$D = \{P_i(\bar{u}_i) ::S_i \mid i\in\{1,\ldots,n\}\}$ and
$\C{\bar{u}_i} \cap free(q_i) = \ES$ for $i\in\{1,\ldots,n\}$.

We need to show how to establish in \emph{ABO} the conclusion
$\HT{p}{S}{q}$ of this rule.  Let
\[
  \mbox{$A := \{i \in \{1, \LL, n\} \mid \mbox{\ the call\ } P_i(\bar{u}_i) \mbox{\ appears in\ } (D \mid S) \}$.}
\]
In the proofs in (\ref{equ:ass1}) and (\ref{equ:ass2}) the assumptions about the 
calls $P_i(\bar{u}_i)$, where $i \not\in A$, are not used, so
\begin{equation}
  \label{equ:Ass1}
  \{\HT{p_j}{P_j(\bar{u}_j)}{q_j} \mid j \in A\} \vdash_{{ABO}, I} \HT{p}{S}{q}  
\end{equation}
and
\[
  \{\HT{p_j}{P_j(\bar{u}_j)}{q_j} \mid j \in A\} \vdash_{{ABO}, I} \HT{p_i}{S_i}{q_i}  
\]
for $i \in \{1, \LL, n\}$. 
  
Fix $i\in\{1,\ldots,n\}$. From $\HT{p_i}{S_i}{q_i}$ we get
successively by the PARALLEL ASSIGNMENT axiom and the COMPOSITION rule
of \cite{Hoa69}
$\HT{p_i}{\bar{u}_i := \bar{u}_i}{p_i}$ and
$\HT{p_i}{\bar{u}_i := \bar{u}_i; \ S_i}{q_i}$, so ---thanks to the
assumption about $free(q_i)$--- we get by the BLOCK rule
\[
  \HT{p_i}{\block{\local \bar{u}_i := \bar{u}_i;\ S_i}}{q_i}.
\]
Hence
\begin{equation}
  \label{equ:Ass2}
  \{\HT{p_j}{P_j(\bar{u}_j)}{q_j} \mid j \in A\} \vdash_{{ABO}, I}   \HT{p_i}{\block{\local \bar{u}_i := \bar{u}_i;\ S_i}}{q_i}.
\end{equation}

By assumption the considered program $(D \mid S)$ is pure, so from
(\ref{equ:Ass1}) and (\ref{equ:Ass2}) for $i \in \{1, \LL, n\}$ we
derive by the RECURSION II rule $\HT{p}{S}{q}$, the conclusion of the
RECURSION rule.
\HB
\end{proof}

We noticed in Section \ref{sec:towards} that each program $(D \mid S)$
can be transformed into a pure program by replacing each non-generic
procedure call $P(\bar{t})$ by the block statement
$\block{\local \bar{u}:=\bar{t}; \ P(\bar{u})}$, where $\bar{u}$ are
the formal parameters of $P$.  Call the effect of all such
replacements the \emph{purification} of $(D \mid S)$.

Denote now by \emph{ABO1} the proof system obtained from \emph{ABO} by
augmenting it by the PROCEDURE CALL rule.

  \begin{theorem} \label{thm:pure1}
  The RECURSION rule is a derived rule in the \textit{ABO1}  proof system.
\end{theorem}
\begin{proof}
  Fix an interpretation $I$. Suppose that we established in
  \emph{ABO1} the premises of the RECURSION rule, i.e.,
  (\ref{equ:Ass1}) and (\ref{equ:Ass2}) for $i \in \{1, \LL, n\}$,
with $\vdash_{{ABO}, I}$ replaced by
  $\vdash_{{ABO1}, I}$. 

Modify these proofs as follows.  Consider an
application of the PROCEDURE CALL rule used in one of these proofs, say
\[
\frac{\HT{p'}{P(\bar{u})}{q'}}
{\HT{p'[\bar{u}:=\bar{t}]}{P(\bar{t})}{q'}}
\]
where $\C{\bar{u}} \cap free(q') = \ES$.

Replace it by the following subproof that makes use of
the PARALLEL ASSIGNMENT axiom and the COMPOSITION and BLOCK
rules:
\[
\begin{array}{l}
\HT{p'[\bar{u}:=\bar{t}]}{\bar{u} := \bar{t}}{p'}, \ \HT{p'}{P(\bar{u})}{q'} \\
[-\medskipamount]
\hrulefill                                                      \\
\HT{p'[\bar{u}:=\bar{t}]}{\bar{u} := \bar{t}; \ P(\bar{u})}{q'} \\
[-\medskipamount]
  \hrulefill \\
\HT{p'[\bar{u}:=\bar{t}]}{\block{\local \bar{u}:=\bar{t}; \ P(\bar{u})}}{q'}
\end{array}
\]

After such changes other rule applications in the assumed proofs
remain valid since
\[
change(D \mid P(\bar{t})) = change(D \mid \block{\local \bar{u}:=\bar{t}; \ P(\bar{u})}).
\]

This way we obtain proofs in \emph{ABO} of the correctness formulas
referring to the purification $(D_{pure} \mid S_{pure})$ of
$(D \mid S)$.  By Lemma \ref{lem:pure} we conclude that
$\vdash_{{ABO}, I} \HT{p}{S_{pure}}{q}$. By replacing in this proof
each introduced subproof of
$\HT{p'[\bar{u}:=\bar{t}]}{\block{\local \bar{u}:=\bar{t}; \ P(\bar{u})}}{q'}$
from the assumption \HT{p'}{P(\bar{u})}{q'} back by
the original application of the PROCEDURE CALL rule
we conclude that $\vdash_{{ABO1}, I} \HT{p}{S}{q}$.
\HB
\end{proof}


As a side effect we obtain the following conclusion.
\begin{corollary}
  The proof system \textit{CBV} sound.
\end{corollary}

\begin{proof}
Consider first the PROCEDURE CALL rule.
Assume that $\C{\bar{u}} \cap free(q) = \ES$ and
that the procedure $P$ is declared by $P(\bar{u}):: S\in D$.
Then by Lemma \ref{lem:sem} and Corollary \ref{cor:ht}$(ii)$
\[
\mbox{$I \models \HT{p}{P(\bar{u})}{q}$ iff $I \models \HT{p}{\bar{u}:=\bar{u}; \ S}{q}$ iff
  $I \models \HT{p}{S}{q}$}
\]
and
\[
\mbox{$I \models \HT{p[\bar{u}:=\bar{t}]}{P(\bar{t})}{q}$ iff $I \models \HT{p[\bar{u}:=\bar{t}]}{\bar{u}:=\bar{t}; \ S}{q}$.}
\]
But by the soundness of PARALLEL ASSIGNMENT axiom
$I \models \HT{p[\bar{u}:=\bar{t}]}{\bar{u}:=\bar{t}}{p}$, so
$I \models \HT{p}{S}{q}$ implies by the soundness of the COMPOSITION
rule $I \models \HT{p[\bar{u}:=\bar{t}]}{\bar{u}:=\bar{t}; \ S}{q}$.
This establishes soundness of the PROCEDURE CALL rule.

The soundness of the RECURSION rule is now a consequence
of the soundness of the proof system \emph{ABO} (established in \cite{ABO09}),
soundness of the PROCEDURE CALL rule, and Theorem \ref{thm:pure1}.

The remaining axioms and proof rules of \emph{CBV} are the same as the corresponding rules of
\emph{ABO}, so the claim follows from the fact
that the proof system \emph{ABO} is sound.
\HB
\end{proof}

\section{Relative completeness}
\label{sec:sandc}

We now prove completeness of \emph{CBV} in the sense of Cook. To this
end, following \cite{Gor75}, we introduce the \emph{most general
  correctness specifications} (\emph{mgcs} in short). In our setup these are
correctness formulas of the form
\[
\HT{\bar{x}=\bar{z}\wedge \bar{u}=\bar{v}}{P(\bar{u})}{\exists \bar{u}:\SP_I(\bar{x}=\bar{z}\wedge \bar{u}=\bar{v},S)}
\]
where 
\begin{itemize}
\item
$P(\bar{u})::S\in D$,
\item
  $\C{\bar{x}}= change(D)\setminus \{\bar{u}\}$,

\item
  $\bar{v}$ and $\bar{z}$ are lists of fresh variables, of the same length as, respectively,
  $\bar{u}$ and $\bar{x}$.
\end{itemize}

In the second condition the removal of the set of variables listed in
$\{\bar{u}\}$ seems at first sight superfluous, since by definition
$change(P(\bar{u})::S) \cap \{\bar{u}\} = \ES$. However, each
procedure can have a different list of formal parameters, so the
variables listed in $\bar{u}$ can also have global occurrences in the
procedure bodies of other procedures, and as a result can appear in
$change(D)$.

The important difference between our mgcs and the ones used in
\cite{Gor75} is that in our case the strongest postcondition of the
call $P(\bar{u})$ is defined in terms of a strongest postcondition of
the body $S$ in which the formal parameters are quantified out.  This
has to do with the fact that we model the call-by-value parameter
mechanism, while ---as already mentioned--- in \cite{Gor75} the
call-by-name parameter mechanism is used. This difference and the fact
that we put no restrictions on the actual parameters leads to a more
complicated argument in the sequel.

Further, note that the sets of variables listed in $\bar{u}, \bar{v}, \bar{x}$ and $\bar{z}$
are mutually disjoint. 
The following technical lemma isolates the crucial implications,
notably one concerning the strongest postcondition 
$sp_I(\bar{x}=\bar{z}\wedge \bar{u}=\bar{v},S)$. They are needed in the
proof of relative completeness.

\begin{lemma} \label{lem:crucial1}
  Suppose that $I$ is such that
  $\cal L$ is expressive relative to $I$ and that
  $I \models \HT{p}{P(\bar{t})}{q}$, where $P$ is declared by
  $P(\bar{u})::S$.

  Let the lists of variables $\bar{v}, \bar{x}$ and $\bar{z}$
  be defined as above and let
\[
  Inv \equiv (p \wedge \bar{t}=\bar{v})[\bar{x}:=\bar{z}].
\]
\begin{enumerate}[(i)]
\item 
$
I \models (Inv\wedge \exists \bar{u}:\SP_I(\bar{x}=\bar{z}\wedge \bar{u}=\bar{v},S))\rightarrow q
$.

\item Suppose that the variables listed in $\bar{v}$ and $\bar{z}$ do
  not appear in $p$. Then
\[
  I \models p \to \exists \bar{v},\bar{z}: (Inv\wedge \bar{x}=\bar{z}\wedge \bar{t}=\bar{v}).
\]
\end{enumerate}
\end{lemma}

\begin{proof}
\mbox{}

\NI
$(i)$
Take a state $\sigma$ such that
\begin{equation}
  \label{equ:u}
\sigma \models_I Inv\wedge \exists \bar{u}:\SP_I(\bar{x}=\bar{z}\wedge \bar{u}=\bar{v},S).  
\end{equation}
It follows that
\begin{equation}
  \label{equ:SP}
\sigma'\models_I \SP_I(\bar{x}=\bar{z}\wedge \bar{u}=\bar{v},S),  
\end{equation}
where $\sigma'=\sigma[\bar{u}:=\bar{d}]$ for some sequence  $\bar{d}$ of elements of the domain of $I$.

By the definition of the strongest postcondition there exists a state $\sigma_0$ such that

\begin{itemize}
\item 
$\sigma_0\models_I \bar{x}=\bar{z}\wedge \bar{u}=\bar{v}$ and

\item $\MSI{S}(\sigma_0) = \{\sigma'\}$.

\end{itemize}

Let  $\sigma''=\sigma[\bar{x}:=\sigma(\bar{z})]$. Assume first that
\begin{equation}
  \label{equ:link}
\sigma''[\bar{u}:=\sigma''(\bar{t})]=\sigma_0.
\end{equation}
Then, given that $\sigma''(\bar{u}) = \sigma(\bar{u})$ and $\sigma'[\bar{u}:=\sigma(\bar{u})] =
\sigma$, Corollary \ref{cor:conc}$(i)$ allows us to derive the following chain of equalities:
\[
\begin{array}{l}
\MSI{P(\bar{t})}(\sigma'') = \\
  \MSI{S}(\sigma''[\bar{u}:=\sigma''(\bar{t})])[\bar{u}:=\sigma''(\bar{u})] = \\
  \MSI{S}(\sigma_0)[\bar{u}:=\sigma''(\bar{u})] = \\
    \MSI{S}(\sigma_0)[\bar{u}:=\sigma(\bar{u})] = \\
  \{\sigma'[\bar{u}:=\sigma(\bar{u})]\} =\\
   \{\sigma\}
\end{array}
\]

Figure \ref{diagram} clarifies the relation between the introduced
states and illustrates the construction of a computation of
$P(\bar{t})$ that starts in the state $\sigma''$ with the final state $\sigma$
and that corresponds to the  above equalities.

\begin{figure}
  \begin{center}
\begin{tikzpicture}
\draw[thick,->] (4,1) -- (0,1) node[anchor=east]{$\sigma''$};
\draw (2,0.8) node { $[\bar{x}:=\sigma(\bar{z})] $};
\draw[thick,->] (4,1.2) -- (4,4) node[anchor=south]{$\sigma'$};
\draw (3.3,2) node[anchor=south]{$[\bar{u}:=\bar{d}] $};
\draw[thick,->] (4.2,4) -- (4.2,1.2) node[anchor=north]{$\sigma$};
\draw (4.3,2) node[anchor=west]{$[\bar{u}:=\sigma''(\bar{u})] $};
\draw[thick,->] (0,1.2) -- (0,4) node[anchor=south]{$\sigma_0$};
\draw (-1,2) node[anchor=south]{$[\bar{u}:=\sigma''(\bar{t})] $};
\draw[thick,->] (0.2,4) -- (3.8,4);
\draw (2,4.2) node {$S$};
\end{tikzpicture}
\caption{Construction of a computation of $P(\bar{t})$ that starts in $\sigma''$.} \label{diagram}
\end{center}
\end{figure}
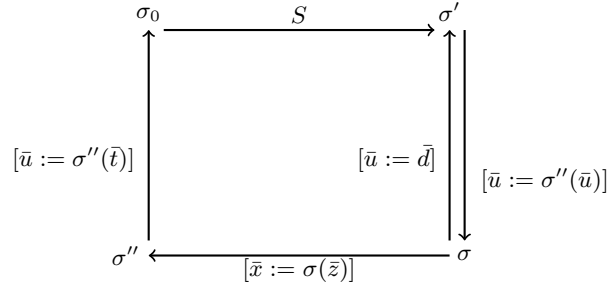

Now, by (\ref{equ:u}) $\sigma\models_I p[\bar{x}:=\bar{z}]$, so
$\sigma'' \models_I p$ by the definition of $\sigma''$.  (This is
actually the contents of the Simultaneous Substitution Lemma in
\cite[p.~50]{ABO09}.) Given that we showed that
$\MSI{P(\bar{t})}(\sigma'') = \{\sigma\}$ we conclude by the
assumption $I \models \HT{p}{P(\bar{t})}{q}$ that
$\sigma \models_I q$.  By the choice of $\sigma$ this establishes the
claim.

It remains to prove (\ref{equ:link}).  First, note the following
consequences of the \textbf{Access and Change} item of Lemma
\ref{lem:sem} that follow from the fact that the variables listed in
$\bar{v}$ and $\bar{z}$ are fresh and that by the choice of $\bar{x}$
we have $change(D) \sse \{\bar{u}\} \cup \{\bar{x}\}$:

\begin{enumerate}[(a)]
\item  $\sigma'(\bar{v}) = \sigma_0(\bar{v})$,
  
\item $\sigma'(\bar{z}) = \sigma_0(\bar{z})$,
  
\item $\sigma'(y) = \sigma_0(y)$, where $y \not\in \{\bar{u}\} \cup \{\bar{x}\}$.

\end{enumerate}

By (\ref{equ:u}) $\sigma \models_I (\bar{t}=\bar{v})[\bar{x}:=\bar{z}]$,
so by the definition of $\sigma''$ (this is, again, an instance of the
Simultaneous Substitution Lemma in \cite[p.~50]{ABO09})
$\sigma'' \models_I \bar{t}=\bar{v}$, i.e.,
\[
  \sigma''(\bar{t})=\sigma''(\bar{v}).
\]
This, item (a) above, the
fact that the sets of variables listed in $\bar{u}, \bar{v}$ and $\bar{x}$
are mutually disjoint, and the definitions of $\sigma''$ and
$\sigma'$ justifies the following chain of equalities:
$$
\sigma''[\bar{u}:=\sigma''(\bar{t})](\bar{u})=
\sigma''(\bar{t})=
\sigma''(\bar{v})=
\sigma(\bar{v}) =
\sigma'(\bar{v})  =
\sigma_0(\bar{v}) =
\sigma_0(\bar{u}).
$$

Next, the same observations and item (b) above justifies the following
chain of equalities:
$$
\sigma''[\bar{u}:=\sigma''(\bar{t})](\bar{x})=\sigma''(\bar{x})= \sigma(\bar{z})=
\sigma'(\bar{z})=\sigma_0(\bar{z})=\sigma_0(\bar{x}).
$$

Finally, take a variable $y\not\in \{\bar{u}\} \cup \{\bar{x}\}$.
Then the above observations and item (c) above justifies the following
chain of equalities:
$$
\sigma''[\bar{u}:=\sigma''(\bar{t})](y)= \sigma''(y)=\sigma(y)=\sigma'(y)=\sigma_0(y).
$$

We thus established (\ref{equ:link}), which concludes the proof.
\III

\NI
$(ii)$
%
%
Take a state $\sigma$ such that $\sigma \models_I p$.  By assumption
the variables listed in $\bar{v}$ and $\bar{z}$ do not occur in $p$, so
$\tau \models_I p$, where $\tau = \sigma[\bar{v}, \bar{z}:= \sigma(\bar{t}), \sigma(\bar{x})]$.

Note that $\tau[\bar{x}:= \tau(\bar{z})] = \tau$, so 
$\tau[\bar{x}:= \tau(\bar{z})] \models_I p$.
By the Simultaneous Substitution Lemma of \cite[p.~50]{ABO09}
$\tau \models_I p[\bar{x}:=\bar{z}]$, i.e.,
\[
\sigma[\bar{v}, \bar{z}:= \sigma(\bar{t}), \sigma(\bar{x})] \models_I p[\bar{x}:=\bar{z}].
\]

Further, for an arbitrary $\sigma$ we have
\[
\sigma[\bar{v}, \bar{z}:= \sigma(\bar{t}), \sigma(\bar{x})]
\models_I
(\bar{t} =\bar{v})[\bar{x}:=\bar{z}]   \wedge \bar{x}=\bar{z} \wedge \bar{t}=\bar{v},
\]
so
\[
\sigma[\bar{v}, \bar{z}:= \sigma(\bar{t}), \sigma(\bar{x})]  
\models_I
(p \wedge \bar{t}
  =\bar{v})[\bar{x}:=\bar{z}] \wedge \bar{x}=\bar{z} \wedge
  \bar{t}=\bar{v}.
\]
This shows that
\[
  I \models p \to \te \bar{v}, \bar{z}: ((p \wedge \bar{t}
  =\bar{v})[\bar{x}:=\bar{z}] \wedge \bar{x}=\bar{z} \wedge
  \bar{t}=\bar{v})
\]
and yields the conclusion by the definition of $Inv$.
\HB
\end{proof}

\begin{theorem}[Completeness] \label{thm:completeness}
The proof system \textit{CBV} is complete in the sense of Cook.
\end{theorem}

As in \cite{Gor75} (see also \cite[pp.~450-452]{Apt81b}) the proof is
based on the following key lemma.  We denote here by $G(D)$ the set of
most general correctness specifications for all procedures declared in
$D$.

\III

\NI
\begin{lemma} \label{lem:crucial}
Let $I$ be an interpretation such that $\cal L$ is expressive relative to $I$.
Suppose that $I \models \HT{p}{T}{q}$. Then
$G(D) \vdash_I \HT{p}{T}{q}$.
\end{lemma}

\begin{proof}
  As in \cite{Gor75} and \cite{ABOG12} we proceed by induction on the
  structure of  the statement $T$,  with two essential cases being different.
  \III

\NI
\emph{Block statements}. Suppose that
$I\models \HT{p}{\block{\local \bar{y} := \bar{t}; \ S}}{q}$.
Let $\bar{y}'$ be some fresh variables 
corresponding to the local variables $\bar{y}$.
Note that by the definition of $change(D \mid S)$
\begin{equation}
  \label{equ:y}
  \{\bar{y}\} \cap change(D \mid \block{\local \bar{y} := \bar{t}; \ S})=\ES.
\end{equation}
So from the soundness of the INVARIANCE rule it follows that
$$I\models \HT{p\wedge \bar{y}=\bar{y}'}{\block{\local \bar{y} := \bar{t}; \ S}}{q\wedge \bar{y}=\bar{y}'}.$$
Consequently,  by the CONSEQUENCE rule
$$I\models \HT{p\wedge \bar{y}=\bar{y}'}{\block{\local \bar{y} := \bar{t}; \ S}}{q[ \bar{y}:=\bar{y}']},$$
because  $q\wedge \bar{y}=\bar{y}'$  implies $q[\bar{y}:=\bar{y}']$.
From Corollary \ref{cor:ht}$(i)$ it then follows that
$$I\models \HT{p\wedge \bar{y}=\bar{y}'}{\bar{y} := \bar{t}; \ S}{q[ \bar{y}:=\bar{y}']}.$$
By the induction hypothesis 
we obtain
$$
G(D)\vdash_I  \HT{p\wedge \bar{y}=\bar{y}'}{\bar{y} := \bar{t}; \ S}{q[ \bar{y}:=\bar{y}']}.
$$
An application of the BLOCK rule then gives
$$
G(D)\vdash_I  \HT{p\wedge \bar{y}=\bar{y}'}{\block{\local \bar{y} := \bar{t}; \ S}}{q[ \bar{y}:=\bar{y}']}.
$$

Thanks to (\ref{equ:y}) we can now apply the SUBSTITUTION rule with
the substitution $[\bar{y}' := \bar{y}]$ and subsequently the
CONSEQUENCE rule to replace $p\wedge \bar{y}=\bar{y}$ by $p$. This
yields
\[
G(D)\vdash_I  \HT{p}{\block{\local \bar{y} := \bar{t}; \ S}}{q}.
\]

\III

\NI
\emph{Procedure calls}. Suppose that $I \models \HT{p}{P(\bar{t})}{q}$,
where $P$ is declared by $P(\bar{u})::S$.

Assume the most general correctness specification
\[
  \{\bar{x}=\bar{z}\wedge \bar{u}=\bar{v}\} P(\bar{u})\{\exists
  \bar{u}:\SP_I(\bar{x}=\bar{z}\wedge \bar{u}=\bar{v},S)\}.
\]

By applying the SUBSTITUTION rule to rename the variables listed in
$\bar{v}$ and $\bar{z}$ we may assume these variables not appear in
$\HT{p}{P(\bar{t})}{q}$.
By the choice of the list $\bar{x}, \bar{z}$ and
$\bar{v}$ none of them contains a variable from $\bar{u}$. So by
the PROCEDURE CALL rule with the substitution
$[\bar{u} := \bar{t}]$ we obtain 
\[
\{\bar{x}=\bar{z}\wedge \bar{t}=\bar{v}\} P(\bar{t})\{\exists \bar{u}:\SP_I(\bar{x}=\bar{z}\wedge \bar{u}=\bar{v},S)\}.
\]

As in Lemma \ref{lem:crucial1} let
\[
  Inv \equiv (p \wedge \bar{t}=\bar{v})[\bar{x}:=\bar{z}].
\]
Note that $free(Inv) \cap \{\bar{x}\}=\ES$ and by definition
$change(D \mid P(\bar{t})) = change(D)\setminus \{\bar{u}\} = \{\bar{x}\}$, so
$free(Inv) \cap change(D \mid P(\bar{t}))=\ES$.
Thus by the INVARIANCE rule 

\begin{equation}
  \label{equ:inv}
\{Inv\wedge \bar{x}=\bar{z}\wedge \bar{t}=\bar{v}\} P(\bar{t})\{Inv\wedge \exists \bar{u}:\SP_I(\bar{x}=\bar{z}\wedge \bar{u}=\bar{v},S)\}.  
\end{equation}

It remains to show that the pre- and postconditions of the above
correctness formula can be replaced, respectively, by $p$ and $q$.

First, we have by Lemma \ref{lem:crucial1}$(i)$
\[
  I \models (Inv\wedge \exists \bar{u}:\SP_I(\bar{x}=\bar{z}\wedge \bar{u}=\bar{v},S))\rightarrow q,
\]
so by the CONSEQUENCE rule we obtain from (\ref{equ:inv}) 
\[
\HT{Inv\wedge \bar{x}=\bar{z}\wedge \bar{t}=\bar{v}}{P(\bar{t})}{q}.
\]

By assumption the variables listed in $\bar{v}$ and $\bar{z}$ do not
appear in $P(\bar{t})$ or $q$, so by the $\te$-INTRODUCTION rule 
\[
\HT{\exists \bar{v},\bar{z}: (Inv\wedge \bar{x}=\bar{z}\wedge \bar{t}=\bar{v})}{P(\bar{t})}{q}.
\]
By Lemma \ref{lem:crucial1}$(ii)$
\[
  I \models p \to \exists \bar{v},\bar{z}: (Inv\wedge \bar{x}=\bar{z}\wedge \bar{t}=\bar{v}),
\]
so we obtain the desired conclusion by the CONSEQUENCE rule.

The remaining cases are as in \cite{Coo78}.
\HB
\end{proof}

\NI
\emph{Proof of the Completeness Theorem \ref{thm:completeness}.}

Suppose, as in Lemma \ref{lem:crucial}, that $I$ is an interpretation
such that $\cal L$ is expressive relative to $I$ and that
$I \models \HT{p}{S}{q}$. To prove $\HT{p}{S}{q}$ in \emph{CBV} we use
the RECURSION rule with $G(D)$ as the set of assumptions in the
subsidiary proofs.  Lemma \ref{lem:crucial} ensures the first premise.
The remaining $n$ premises also follow by this lemma.  Indeed, suppose
  \[
    G(D) =   \{\HT{\bar{x}_i =\bar{z}_i \wedge \bar{u}_i =\bar{v}_i}{P_i(\bar{u}_i)}{\exists \bar{u}_i: \SP_I(\bar{x}_i=\bar{z}_i\wedge \bar{u}_i=\bar{v}_i,S_i)} \mid i \in \{1, \LL, n\}\},
  \]
where $D = \{P_i(\bar{u}_i) ::S_i \mid i\in\{1,\LL,n\}\}$.


Choose an arbitrary $i \in \{1, \LL, n\}$. By the definition of the strongest postcondition
  \[
    I \models \HT{\bar{x}_i =\bar{z}_i \wedge \bar{u}_i =\bar{v}_i}{S_i}{\SP_I(\bar{x}_i=\bar{z}_i\wedge \bar{u}_i=\bar{v}_i,S_i)},
\]
so
  \[
    I \models \HT{\bar{x}_i =\bar{z}_i \wedge \bar{u}_i =\bar{v}_i}{S_i}{\exists \bar{u}:\SP_I(\bar{x}_i=\bar{z}_i\wedge \bar{u}_i=\bar{v}_i,S_i)}
  \]
by the soundness of the CONSEQUENCE rule. Hence by Lemma \ref{lem:crucial}
\begin{equation}
  \label{equ:ht}
G(D) \vdash \HT{\bar{x}_i =\bar{z}_i \wedge \bar{u}_i =\bar{v}_i}{S_i}{\exists \bar{u}:\SP_I(\bar{x}_i=\bar{z}_i\wedge \bar{u}_i=\bar{v}_i,S_i)}\}.
\end{equation}
We conclude now  $\vdash_{I} \HT{p}{S}{q}$ by the RECURSION rule.
\HB

\section{Length of proofs}
\label{sec:length}

In Section \ref{sec:towards} we implicitly indicated that there is a
linear bound on the length of proofs in the proof system \emph{CBV}.
We can make this claim formal by reassessing the proof of the
completeness result given in the previous section.

First, to each program $(D \mid S)$ we assign inductively a number
$l(D \mid S)$ as follows:
 
\begin{itemize}
\item $l(skip) := 1$,
  
\item $l(\bar{x}:=\bar{t}) := 1$,

\item $l(P(\bar{t})) := 1$,

\item $l(S_1 ; \ S_2) := l(S_1)+l(S_2)+1$,

\item $l(\ITE{B}{S_1}{S_2}) := l(S_1)+l(S_2)+1$,

\item $l(\WDD{B}{S}) := l(S)+1$,

\item $l(\block{\local \bar{x}:=\bar{t}; \ S}) := l(\bar{x}:=\bar{t}; \ S)+1$,

\item $l(P(\bar{u})::S) := l(S)$,

\item $l(\{P(\bar{u})::S)\} \cup D) := l(P(\bar{u})::S) + l(D)$,

\item $l(D \mid S ) :=  l(D) + l(S)$.

\end{itemize}
Further, denote by
\begin{itemize}
\item $n_a(S)$ the number of the parallel assignments occurring in $S$,

\item $n_b(S)$ the number of the block statements occurring in $S$,  
  
\item $n_p(S)$ the number of the procedure calls occurring in $S$,
  
\item $n_w(S)$ the number of the \textbf{while} loops occurring in $S$,
\end{itemize}
and let
\[
  m(S) := l(S) + n_a(S) + 3n_b(S) + 6n_p(S)  + n_w(S).
\]
Note that $l(D \mid S)$ is smaller than the length of $(D \mid S)$ viewed as a string
of characters and that $m(S) < 13l(S)$.

Further, given a set of correctness formulas $\Phi$ and an interpretation $I$ denote by
\[
\Phi \vdash^k_{I} \HT{p}{S}{q}
\]
the fact that there exists a proof of $\HT{p}{S}{q}$ in \emph{CBV}
from $\Phi$ and the set of all assertions true in $I$ that consists of
at most $k$ applications of axioms and proof rules.  We can now refine
Lemma \ref{lem:crucial} as follows.

\begin{lemma} \label{lem:length}
Let $I$ be an interpretation such that $\cal L$ is expressive relative to $I$.
Suppose that $I \models \HT{p}{T}{q}$. Then
$G(D) \vdash^{m(T)}_I \HT{p}{T}{q}$. 
\end{lemma}
\begin{proof}
  We reassess the proof of Lemma \ref{lem:crucial} and count the
  number of applications of axioms and proof rules. For the case of a
  block statement note that
  \[
    m(\block{\local \bar{y} := \bar{t}; \ S_0}) = m(\bar{y} := \bar{t}; \ S_0) + 3
  \]
  and that the correctness proof of the former statement was obtained
  from a correctness proof of the latter statement by 3 rule
  applications.  In the case of a procedure call it suffices to check
  that the correctness proof consisted of 6 rule applications.

  The remaining cases follow by an inspection of the relative
  completeness proof given in \cite{Coo78}. For example, for the
    \textbf{while} statement $\WDD{B}{S}$ 
  \[
    m(\WDD{B}{S}) = m(S)+2
  \]
  and the correctness proof of the former statement was obtained from
  a correctness proof of $S$ by an application of the WHILE rule of
  \cite{Hoa69} followed by an application of the CONSEQUENCE rule.
  \HB
\end{proof}

This allows us to sharpen the Completeness Theorem \ref{thm:completeness} as follows.
\begin{theorem} \label{thm:length}
Let $I$ be an interpretation such that $\cal L$ is expressive relative to $I$.
  Suppose that $I \models \HT{p}{T}{q}$.
 Then  $\vdash^{\mathcal{O}(l(D \mid T))}_I \HT{p}{T}{q}$.
\end{theorem}
\begin{proof}
  We reassess the proof of the Completeness Theorem
  \ref{thm:completeness}. The correctness formula $\HT{p}{T}{q}$ was
  established using the RECURSION rule.  By Lemma \ref{lem:length} to
  prove the premises of this rule it takes
at most $m(T) + \Sigma_{i=1}^{n} m(S_i)$ steps, where $n$ is the
number of procedure declarations in $D$ and $S_1, \LL, S_n$ are the
bodies of the procedures declared in $D$.  But $m(S) < 13l(S)$ and
$l(D \mid T) = l(T) + \Sigma_{i=1}^{n} l(S_i)$, so the claim follows.
\HB
\end{proof}

\section{Final remarks} \label{sec:final}

\subsection{Related work}

Let us compare now our relative completeness result with the work
reported in the literature.

As already mentioned in the Introduction, the first sound and
relatively complete proof system for programs with local variables and
recursive procedures was provided in \cite{Gor75}. But the paper
assumed the call-by-name parameter mechanism and, as explained in
\cite[pp.~459-460]{Apt81b}, dynamic scope was assumed.
The relative completeness
result also assumed some restrictions on the actual parameters in the
procedure calls that were partly lifted in \cite{CartwrightOppen81}.
However, the restriction that global variables do not occur in the actual
parameters is still present there.

In \cite{Bak79} and in more detail in \cite[Section 9.4]{Bak80} a
proof system concerned with the recursive procedures with the
call-by-value and call-by-variable (present in {\sc Pascal}) parameter
mechanisms in presence of static scope was proposed and its soundness
and relative completeness was proved.  However, the correctness of
each procedure call had to be proved separately.  Further, the
relative completeness was established only for the special case of a
single recursive procedure with a single recursive call. For the case
of two recursive calls a list of 14 cases was presented in
\cite[Section 9.4]{Bak80} that should be considered, but without the
corresponding proofs.  The main ideas of this proof were discussed in
\cite{Apt81b}.  The case of a larger number of recursive calls and a
system of mutually recursive procedures were not analyzed because of
the complications resulting from an accumulation of cases resulting
from the use several variable substitutions.

In \cite{Ohe99} a sound and relatively complete proof system was
discussed for a programming language with local and global variables
and mutually recursive procedures with the call-by-value mechanism,
that allows for both static and dynamic scope. The proofs were
certified in the Isabelle theorem prover. The details are very
sketchy, but from the presentation it is clear that also here
correctness of each procedure call has to be dealt with separately.

Further, in the presentation the assertions are identified with their
meaning.  This avoids various complications concerned with possible
variable clashes (in particular, it obviates the use of the
SUBSTITUTION rule). This simplifies the arguments about soundness and
relative completeness, but results in a highly inconvenient approach
to program verification.

To see the difference, take the correctness formula
$\{\T\} x:=0\{x=0\}$.  In Hoare's logic it follows directly by the
ASSIGNMENT axiom
and a trivial application of the CONSEQUENCE rule.  In the approach of
\cite{Ohe99} (that is common for papers on Hoare's logic in the
framework of the Isabelle system) one abstracts from the underlying
logic and uses following assignment axiom:
\[
  \{\lambda\sigma. P(\sigma[x:=0])\} x:=0 \{P\},
\]
for any $P\in \Sigma \rightarrow \mathbb{B}$, where $\Sigma$ denotes
the set of states and
$\mathbb{B} = \{\mathbf{true}, \mathbf{false}\}$.  Semantically, the
substitution $[x :=0]$ is then represented by the function
$\lambda\sigma.\sigma(x)=0$.  Instantiating the above assignment axiom
we obtain
\[
\{\lambda\sigma. \sigma[x:=0](x)=0\} x:=0 \{\lambda\sigma.\sigma(x)=0\}.
\]
Unraveling the update function yields
$\{\lambda\sigma. \T\} x:=0 \{\lambda\sigma.\sigma(x)=0\}$.  Clearly,
this is very cumbersome in practice, also if one wants to use
  assertions as loop invariants or to specify procedures' behaviour.
  In contrast, Hoare's logic is a convenient formalism for such
  purposes.

As already mentioned in Section \ref{sec:pure}, in \cite{ABO09} and
\cite{ABOG12} a proof system \emph{ABO} was proposed for the
programming language here considered and shown to be sound and
relatively complete. (It should be mentioned that the notion of
expressiveness used in \cite{ABOG12} refers to the definability of the
weakest precondition as opposed to the strongest postcondition, as
used here.)  However, as already explained in Section
\ref{sec:towards}, correctness of each procedure call has to be dealt
with \emph{ABO} separately and as a result the correctness proofs in
this proof system are only quadratic in the length of the programs,
even in the presence of just one recursive procedure.

In \cite{ABO09} an attempt was made to circumvent this inefficiency by
proposing a proof rule similar to our PROCEDURE CALL rule so that one
could use a recursion rule only for generic calls. However, even in
the absence of recursive calls, the proposed rule was not applicable
to all procedure calls due to the imposed restrictions.  More
specifically, the proof rule had the following form:
\III

\NI
PROCEDURE CALL I
\[
\frac{\HT{p}{P(\bar{x})}{q}}
{\HT{p[\bar{x}:=\bar{t}]}{P(\bar{t})}{q[\bar{x}:=\bar{t}]}}
\]
where $var(\bar{x}) \cap var(D) = 
var(\bar{t}) \cap change (D) = \ES$.
\III

It was then noted there that the stipulated conditions make it
impossible to derive the correctness formula (\ref{equ:sum}) of
Example \ref{exa:1}$(ii)$ of Section \ref{sec:proof} because the actual
parameter contains a variable present in $change(D)$.  Example
\ref{exa:1} explains how this problem can be resolved using the
PROCEDURE CALL rule in combination with the SUBSTITUTION rule.

We conclude  that the relative completeness result presented here is new.
It is useful to discuss how we dealt with the complications
encountered in the reported papers. 

In the relative completeness proofs in \cite{Coo78}, \cite{Gor75} and
\cite{Bak80} the main difficulties had to do with the local
variables, the use of which led to some variable renamings, and the
clashes between various types of variables, for example formal parameters and
global variables, that led to some syntactic restrictions.

In fact, the original paper of Cook, \cite{Coo78} contains an error
that was corrected in \cite{Coo81}. It is useful to discuss it (it was
actually pointed out by one of us) in some detail.
In the semantics adopted in \cite{Coo78} local variables were modelled
using a stack in which the last used value was kept on the stack and
implicitly assigned to the next local variable. As a result the
following correctness formula is true in any interpretation:
\[
 \HT{\T}
  {\block{\local x; \ x := 1}; \ \block{\local x; \ y := x}}
  {y = 1}.   
\]

However, there is no way to prove it. In the paper of Gorelick
\cite{Gor75} this error does not arise, since all variables are
explicitly initialized to a given in advance value, both in the
semantics and in the proof theory (by adjusting the precondition of
the corresponding declaration rule).

However, this problem does arise in the framework of
\cite{CartwrightOppen81}, where the authors write on page 371:

\begin{quote}
  ``As Apt (personal communication) has observed, this [the
  declaration] rule is incomplete because it does not allow one to
  deduce the values of new variables on block entry. There are several
  possible solutions to this technical problem but they are beyond the
  scope of this paper.''
\end{quote}

In our framework this problem cannot occur because of the explicit
initialization of the variables in the block statement. However, our
initialization is more general than that of Gorelick. For example, in
the statement $\block{\local u:=u; \ S}$ the local variable $u$ is
initialized to the value of the global variable $u$.  Consequently,
according to our semantics, and also the semantics used in \cite{Coo78},
the following correctness formula is true in any interpretation:
\begin{equation}
  \label{equ:local}
 \HT{\T}  
{\block{\local u:=u; \ x:=u}; \ \block{\local u:=u; \ y:=u}}
{x=y}.
\end{equation}
This correctness formula cannot be proved in the proof system used \cite{Coo78}.
However, we can prove it in our proof system as follows.
As shown in Example \ref{exa:1}$(i)$
\[
  \HT{\T}{\block{\local u:=u; \ x:=u}}{x=u}
\]
and
\[
  \HT{\T}{\block{\local u:=u; \ y:=u}}{y=u}.
\]
Applying to the last correctness formula
the INVARIANCE rule we get
\[
  \HT{x=u}{\block{\local u:=u; \ y:=u}}{x=u \land y=u}.
\]
(\ref{equ:local}) now follows by the COMPOSITION and
CONSEQUENCE rules.

\subsection{Summary of our approach}

We established in Section \ref{sec:sandc} the Completeness Theorem
\ref{thm:completeness} for all programs defined in Section
\ref{sec:syntax}. As explained in Section \ref{sec:semantics}
semantics of these programs follows dynamic scope. Further, for the
natural subset of clash-free programs, introduced in Definition
\ref{def:clash-free}, static and dynamic scope coincide.  In other
words, for clash-free programs we have a sound and relatively complete
proof system in presence of static scope and in which correctness
proofs are linear in the length of the programs.

It is also useful to note that the proof of relative completeness
proceeds by induction on the structure of the programs and for
clash-free programs the inductive reasoning remains within the set of
these programs.  The only place in the proof that requires some
explanation is the case of the block statements, for which it suffices
to note that if the program
$(D \mid \block{\local \bar{u}:=\bar{t}; \ S})$ is clash free, then so
is $(D \mid \bar{u}:=\bar{t}; \ S)$.

Let us discuss now how we succeeded to circumvent the complications
reported above.  We achieved it by various design decisions concerning
the syntax and semantics, which resulted in a simple proof
system. Some of these decisions were already taken in \cite{ABO09} and
used in \cite{ABOG12}.  More precisely, we used there the block
statement in which the local variables can be initialized to any
expression. This, in conjunction with the parallel assignment, allowed
us to model procedure calls in a simple way, by inlining.  In the
present paper the problems concerning variable clashes were taken in a
more general way by singling out the class of clash-free programs, for
which the provided semantics ensures static scope (that coincides then
with the dynamic scope).  This class of programs was first considered
in \cite{CartwrightOppen81}, where it is introduced on page 372 in a
somewhat informal way:

\begin{quote}
''[$\dots$] we assume that our PASCAL subset [$\dots$] requires that
the global variables accessed by a procedure be explicitly declared at
the head of the procedure and that these variables be accessible at
the point of every call.''
\end{quote}

Crucially, our semantics the block statement does not require any
variable renaming (see the \textbf{Block} item of Lemma
\ref{lem:sem}).  This leads to a simple semantics of the
procedure calls, without any variable renaming either
(see the \textbf{Inlining} item of Lemma \ref{lem:sem}).

As a result, in contrast to all works in the literature, our
proof rule dealing with local variables (the BLOCK rule) uses no
substitution. In contrast, in \cite{Hoa71} the substitution
is applied to the program, while in \cite{Coo78} and \cite{Gor75} it
is applied to the assertions. Further, in contrast to \cite{Bak80},
our RECURSION rule does not involve any substitution in the procedure
body.  This allowed us to circumvent the troublesome combinatorial
explosion of the cases encountered in the relative completeness proof
of \cite{Bak80}.

However, the key improvement over \cite{ABO09} and \cite{ABOG12}, is that our
RECURSION rule deals only with the generic calls and ---crucially---
the PROCEDURE CALL rule does not impose any restrictions on the actual
parameters.  The latter is in contrast to all works that dealt with
the call-by-name parameter mechanism.  Thanks to this improvement, in
contrast to the above two works, in our proof system the correctness
proofs are linear in the length of the program.

\subsection{Conclusions}

The already mentioned work of \cite{Cla79} led to a research on proof
systems for programming languages in which in presence of static scope
both nested declarations of mutually recursive procedures and
procedure names as parameters of procedure calls were allowed. In
particular in \cite{Old84}, \cite{DJ83} and \cite{GCH89} sound and
relatively complete Hoare-like proof systems were established, each
time under different assumptions concerning the assertion language. In
all these papers the call-by-name parameter mechanism was used
and the actual parameters were variables.
It would be interesting to see how to modify these proofs to
programming languages in which the call-by-value parameter mechanism
is used instead of the call-by-name.

Finally, it would be useful to extend (in an appropriate sense) the
results of this paper to total correctness of programs.  In this
context we should mention \cite{AB90}, where a sound and relatively
complete (in an appropriate sense) proof system for total correctness
of programs in presence of recursive procedures was provided.
However, in the paper only procedures without parameters were
considered.


\subsection*{Acknowledgement}

We thank Ernst-R\"{u}diger Olderog and Stijn de Gouw for helpful
comments about the contributions of a number of relevant references.

\bibliographystyle{abbrv}
\bibliography{ao,new,abo}

\end{document}